\documentclass[a4paper,12pt]{article}

\usepackage{amssymb}
\usepackage{latexsym}
\usepackage{amsmath}
\usepackage{amsthm}
\usepackage{graphicx}
\usepackage{fancyhdr}
\usepackage{setspace}
\usepackage[top=1cm, left=2cm, right=2cm, bottom=2cm]{geometry}
\def\P{\mathbb{P}}
\def\E{\mathbb{E}}

\def\abs#1{{|{#1}|}} %absolute value
\def\floor#1{\lfloor{#1}\rfloor} %floor
\def\c#1{{\cal #1}}
\def\r#1{{\rm #1}}

\DeclareMathOperator{\var}{var}

\theoremstyle{plain}
\newtheorem{theorem}{Theorem}[section]

\newtheorem{lemma}[theorem]{Lemma}

\theoremstyle{definition}
\newtheorem{definition}[theorem]{Definition}

\def\cover{
\centerline{\large Dynamic monopolies with randomized starting configuration}  
\vskip.5cm
\centerline{Tom\' a\v s Kulich}
\vskip1cm
\centerline{Department of Computer Science}
\centerline{Faculty of Mathematics, Physics and Informatics}
\centerline{Comenius University}
\centerline{Mlynsk\' a dolina, 84248 Bratislava, Slovak Republic}
\centerline{\small \emph{E-mail:} kulich@dcs.fmph.uniba.sk}
\vskip1cm
\vskip1cm

\centerline{\textbf{Abstract}}
Properties of systems with majority voting rules have been exhaustingly studied. In this work we focus on the randomized case - where the system is initialized by randomized initial set of seeds. Our main aim is to give an asymptotic estimate for sampling probability, such that the initial set of seeds is (is not) a dynamic monopoly almost surely. After presenting some trivial examples, we present exhaustive results for toroidal mesh and random 4-regular graph under simple majority scenario.
}

\begin{document}
\cover

%\begin{keywords}
%Random graph, random subgraph, hypercube, diameter
%\end{keywords}
%}}} 

\section{Introduction} %{{{
Idea of \emph{majority voting} is commonly used to resolve many problems related
to achieving consistency between different parts of distributed computation. For
example, majority voting is used to preserve data consistency when updating
copies of the same data. Also, it is quite common to use majority voting to
resolve inconsistencies in distributed database management. Majority based
systems were also successfully used by Agur to examine the plasticity and
precision of the immune response in \cite{1}. 

The model for the system is as follows: Let $G=(V,E)$ be simple undirected graph
of size $n$. Every vertex has its color which is either black or white, and
represents the state of the node (for example black = faulty, white = not
faulty). By $S$ we shall denote the set of black vertices at the beginning of
the process. These vertices are also called \emph{seeds}. By evolution of such
system (or, the \emph{coloring process}) we shall mean the synchronous process
where in each step each vertex adjusts its color according to colors of its
neighbors and its internal contamination and decontamination rules.
(De)contamination rules determine under what configuration of its neighbors the
white (black) vertex turns black (white). In this work we focus on the case when
there is no decontamination rule, and the contamination rule is the simple
majority rule. This means that the white vertex turns black if at least half of
its neighbors are black and there is no possibility for black vertex to turn
white. Set of seeds $S$ is called \emph{dynamic monopoly}, (or
\emph{dynamo} for short) if the corresponding coloring process leads to
monochromatic black graph. For more rigorous definition, see section
\ref{subsec:torus_preliminaries}.

A significant attention was given to this model, and many interesting results
were obtained. Probably the most basic (but certainly not trivial) question asks
for determining the minimal cardinality of a dynamo on a fixed graph $G$
\cite{FGS01, FLLPS04}. Another interesting parameter of a dynamo besides its
cardinality is the time that is needed for contamination to spread. This was
analyzed for the first time in the literature in \cite{FGS01, FLLPS04}. Several
works are related to more advanced topics such as decontamination of the system
by external agents \cite{FNS05}. Finally, \cite{KR01} defines the terminology
of immune subgraphs and asks how the immune subgraph of a certain graph looks
like.

All these tasks were solved for small class of graphs. Close attention is given
to the ring and its several modifications \cite{FGS01}, as well as to the torus
and its modifications \cite{FLLPS04} (toroidal mesh, torus cordalis, torus
serpentinus). Many results exist also for hypercube and binary tree.

Other possible questions ask about the minimal cardinality of a dynamo on
arbitrary graphs. It is proved that on general directed graph the minimal dynamo
has at most $0.727|V|$ vertices. It is likely that this result will be improved.
On the other hand, in the case of general undirected graphs the minimal dynamo
consists of at most $\floor{V/2}(+1)$ vertices (depending on the exact
specification of contaminating rules), this result is proved to be sharp
\cite{16}.

There exist some interesting results about such systems even in the
''most general'' case, where contamination rule is that the vertex is
contaminated if at least fraction $\alpha$ of its neighbors are black
\cite{dynn1, dynn2, dynn3}. These works ask under what condition at least some
fraction $\delta$ of all vertices is turned black.

 Finally, we mention several works that we find most related to the paper. For a
randomly chosen set of initial black vertices, Gleeson and Cahalane \cite{dynn2}
gave an exact formula for the expected fraction of black vertices at the end in
tree-like graphs. In \cite{dynn1, dynn4} authors gave their estimate for minimal
number of black vertices needed for re-coloring of at least fraction $\delta$ of
all vertices on Erd\H{o}s-R\' enyi random graph. For more thorough survey, see \cite{Floc09, Pel02}.

 Motivation for studying coloring process induced by random set of seeds is
quite straight-forward. For example, considering vertices as computing nodes,
each node can fail with probability $p$, independent on failure of other nodes.
Although this looks like a very common scenario, there exists only few results
about systems initialized with random initial coloring. Given graph $G$ We
consider random initial set of seeds $S^p=S^p(G)$ as the set containing any
vertex of $G$ with probability $p$ (independently on other vertices). Naturally,
for every fixed graph $G$ this gives us some probability that the random set of
seeds is a dynamo.  Determining these probabilities analytically is quite hard
(if not impossible) and moreover, it can be done numerically with sufficient
accuracy. Therefore, we shall try to obtain asymptotic results of the form:
assuming 4-regular graph with $n$ vertices, random set of seeds $S^{0.12}$ is
dynamo with high probability (w.h.p.). On the other hand, $S^{0.10}$ is not a
dynamo (w.h.p.).

In many situations, to determine the minimal value of $p$ that $S^p$ is (w.h.p)
dynamo is trivial. For example, assuming Erd\H{o}s-R\' enyi random graph
$G(n,p')$ for fixed $p'$, $S^p$ forms a dynamo almost surely (a.s.) if $p>1/2 +
\varepsilon$ (where $\varepsilon>0$ is arbitrarily small constant) and does not
form a dynamo a.s. if $p<1/2 - \varepsilon$ ( this follows from Chernoff
bounds and Markov inequality). If in the same model we allow $p'$ to be
dependent on $n$, but not to be significantly decreasing, the same results can
be easily derived. Finally, if $p'$ decreases significantly with $n$ (for
example $p=c/n$), the graph contains (with probability tending to some $\gamma >
0$) some isolated vertices. Therefore until $p$ is not quite close to $1$ (so
close that all isolated vertices are a.s. seeds) we can not say that $S^p$ is w.h.p.
dynamo. Another trivial example is the toroidal mesh under strict majority
contamination rule (that is the white vertex turns black only if it has at least
three black neighbors). In this case, any square of size $2 \times 2$ forms an
immune subgraph (that is such subgraph that turns black only if some of its
vertices are seeds).  Therefore, the sampling probability that would form
dynamic monopoly must prevent all such squares from being colored entirely
white. Once again, this can be done only for $p$ very close to $1$. 

However, the motivation presented in previous text leads to considering $p$ to
be a fault probability of a node in the network (or something alike). Therefore
we can assume $1-p$ not to be very small. This makes the cases where $p
\rightarrow 1$ non-realistic.

In this work we shall examine dynamic monopolies with random initial condition
on two types of underlying graphs. In the section \ref{sec:torus} we focus on
the toroidal mesh. The results are quite surprising - $S^p$ containing only a
fraction $o(1)$ of all the vertices a.s. form a dynamo. To be more concrete, we shall
show that there exist constants $\alpha, \beta$ such that if $p>\alpha /
\ln(n)$, then $S^p$ a.s. is a dynamo. Similarly, when $p<\beta / \ln(n)$, then
$S^p$ a.s. is not a dynamo. At the end of the section we present our attempt to
solve the problem numerically and we show ''measured'' $\alpha$ and $\beta$. As
it was said earlier, to form a dynamo on random graph $G(n,p')$ we need $p$ to
be close to $1/2$. This makes the results about  the toroidal mesh even more
interesting. The natural question arises, whether the reason that such low
values of $p$ are needed to form a dynamo comes from specific topology of
toroidal mesh, or whether it is just implied by the fact, that the degree of the
vertices is constant and low. This is the motivation for the section
\ref{sec:regular}, where we investigate the same questions but on random
$4$-regular graphs, that is the graphs with all the vertices having degree $4$.
We show that if $p \geq 0.12$, then the random initial coloring a.s. forms a
dynamo. Similarly, if $p \leq 0.10$, then the random initial coloring a.s. does
not form a dynamo.
%}}}

\section{Toroidal mesh} \label{sec:torus} %{{{ preliminaries

\subsection{Preliminaries} \label{subsec:torus_preliminaries}
By \emph{toroidal mesh of size $n$} we shall consider an undirected graph
$G=G(V,E)$ consisting of $n^2$ vertices labeled as $V[i,j]$ for $0 \leq i,j \leq
n$. The set of edges consists of all pairs $(V[i,j],V[(i+1)\mod n,j])$ and
$(V[i,j],V[i,(j+1)\mod n])$ for all $0 \leq i,j < n$. By \emph{rectangle} of
size $w \times h$ located at position $x, y$ we denote the subgraph of $G$
induced by vertices $V[i,j]$ for $x \leq i < x+w$ and $y \leq j < y+h$. Such
rectangle will be denoted as $V[x,y,w,h]$. The set of all vertices of the
rectangle $R$ having degree (in $R$) lower than $4$ is called
\emph{circumference} of $R$. Rectangles with size $w \times 1$ or $1 \times h$
are called \emph{lines}, rectangles with $w=n$ or $h=n$ are called vertical or
horizontal \emph{stripes}. Note that circumference of a stripe consists of two
closed lines.

Let $\psi: V \rightarrow \mathbb{N}$ and $S \subseteq V$ be the set of seeds. By
\emph{coloring process} induced by $S$ we shall denote progression of sets
$\c{N}(S,G,\psi)=S_0, S_1, \ldots$ where $S_0=S$ and $S_i \subseteq V$
represents the set of black vertices in the $i$-th step of the coloring process.
The relation between $S_i$ and $S_{i+1}$ is that $S_{i+1}=S_i \cup B_{i+1}$,
where $B_{i+1} \subseteq V$ is the set of all vertices $v$ that are white in the
$i$-th step of the coloring process (that is, $v \notin S_i$) and they have at
least $\psi(v)$ neighbors that are black in the $i$-th step of the coloring
process (that is, they belong to $S_i$). Furthermore, we shall set $B_0=S$ and
we shall say that vertices from $B_i$ are \emph{turned black in time} $i$.
Vertex $v$ is said to be turned black (by the coloring process) if it is turned
black in some time.  Since the object of our interest is the $4$-regular graph
under the simple majority scenario, we shall put throughout the paper
$\psi(v)=2$ for all $v \in V$. If for some $i$, the coloring $S_i$ contains all
vertices $V$, then we say that the corresponding $S$ is \emph{dynamic monopoly}
or \emph{dynamo} (with respect to $G$ and $\psi$). 

Vertex that belong (does not belong) to set of seeds $S$ shall be called
$S$-vertex (non-$S$-vertex). Similarly, rectangle (line) consisting of $S$ vertices
(non-$S$-vertices) is called $S$-rectangle ($S$-line, non-$S$-rectangle, non-$S$-line).

We want to analyze the coloring process with random initial condition. For this
purpose, we define the random set of seeds $S^p=S^p(G)$. This is obtained as follows:
Every vertex $v$ of $G$ is contained in $S^p$ with probability $p$, independently on
the other vertices. All possible set of seeds form together with
adequate probabilistic measure a probabilistic space of random set of seeds
$\c{S}^p=\c{S}(G)^p$. We put $S^p \in \c{S}^p$ throughout the text. Naturally, probability
$p$ can depend on $n$. Since our aim is to obtain an asymptotic description of
the behavior of the system, all limits we use are for the case $n \rightarrow
\infty$. We shall say that event $A$ happens \emph{with high probability}
(w.h.p.), if $\lim_{n \rightarrow \infty} \P(A) \rightarrow 1$. Alternatively we
shall say that $A$ happens \emph{almost surely} (a.s.).

For abbreviation we shall put $l=l(n)=\floor{\ln(n)}$ and $q=1-p$.

%}}}

\subsection{Lower bound} \label{section lower} %{{{ uvodne kecy k lower boundu
In this section we show that unless $p$ is high enough, there w.h.p. exists
covering of $G$ by \emph{cages} (for the definition of cage see below), that prevents the black
vertices from spreading. The result of this section is stated in the following
lemma and proved at the end of the section.

\begin{theorem} \label{not black}
If $p \leq 0.012 / \ln(n)$, then the random set of seeds $S^p$ w.h.p. is not
a dynamo on graph $G$.
\end{theorem}

\noindent In this section we shall assume $n$ to be divisible by $2l$. We shall
show that this is only a technical complication at the end of this section.

Let $G$ be a toroidal mesh of size $n \times n$ and let $S^p \in \c{S}^p$ be
random set of seeds. From this graph we can derive corresponding
graph $G'$ and corresponding set of seeds $S'$ as follows: $G'$ is a toroidal mesh of a size $n/2
\times n/2$. We shall say that $V'[x,y]$ corresponds to foursome of vertices
$V[2x,2y], V[2x+1,2y], V[2x,2y+1], V[2x+1,2y+1]$. Vertex $v' \in V'$ belongs to
$S'$ if at least one of its corresponding vertices is $S^p$-vertex. 
Note that probability $p'$ that the vertex in $G'$ is $S'$-vertex is 
$$p'=1-(1-p)^4 \leq 4p$$
and for $p \rightarrow 0$ we get $p' \sim 4p$. %}}}

\begin{definition} %{{{ definicia cage
By \emph{cage} we denote the circumference of some rectangle $R$ in $G'$
consisting of non-$S'$-vertices only. In more detail, for any integers
$x,y$, $w \leq n$, $h \leq n$ if the vertex of $G'$ on the coordinates $[x+i,y],
[x+i,y+h], [x,y+j], [x+w,y+j]$ is not member of $S'$ for all $0\leq i < w$ and
$0 \leq j < h$, then these vertices form cage in $G'$ and vertices corresponding
to these form cage in $G$. By the \emph{interior of the cage} we denote those
vertices of $R$ that do not lie on the circumference.  
\end{definition} 

%}}}

\begin{lemma} \label{cage} %{{{
If there exists a set $\c{K}$ of cages such that
\begin{itemize}
\item every $v' \in S'^p$ belongs to the interior of some cage $K \in
\c{K}$;
\item for $K_i, K_j \in \c{K}$ and $K_i \neq K_j$ the interiors of the cages $K_i$
and $K_j$ are disjoint,
\end{itemize}
then $S^p$ is not dynamic monopoly on $G$.
\end{lemma}
\begin{proof}
We can assume all seeds to be ''trapped'' inside cages which they can not
escape from. More rigorous proof can be done by induction on the number of
steps of the coloring process.
\end{proof} %}}}

%{{{ definicia dobreho rectanglu
\noindent Let us consider rectangle $R=V'[x,y,w,h]$ of $G'$. We shall call this rectangle
\emph{good}, if 

\begin{itemize}
\item $w \geq h$ and there exists a non-$S'$-line $V'[x,y+i,w,1]$ for some $0 \leq i
< h$ of $R$,  or
\item $w < h$ and there exists a non-$S'$-line $V'[x+i,y,1,h]$ for some $0 \leq
i < w$ of $R$.
\end{itemize}

If none of the above holds, we shall call $R$ \emph{bad}.
%}}}

\begin{lemma} \label{line} %{{{
Let $\alpha$ be arbitrarily low positive number, $\beta$ the root of the 
equation $$\ln(1-e^{-12\beta})=-2$$ and let $p$ be equal to $(\beta-\alpha)/ \ln(n)$. Then $G'$
a.s. does not contain bad rectangle of size $l \times 3l$ or $3l \times l$.  
\end{lemma}

\begin{proof}
As was mentioned before, $p' \sim 4p=4(\beta-\alpha)/\ln(n)$
The probability that the rectangle with width $w$ and height $h$, $w>h$ is bad
equals to:
$$p_\r{one~rec}=(1-(1-p')^w)^h.$$
There are $2n^2$ different rectangles of given dimensions, therefore the
probability that there will be at least one bad rectangle of given dimensions
can be upper-bounded using
union bound by:
$$p_\r{all~recs} \leq 2n^2 \cdot p_\r{one~rec}.$$
By letting $w=3l$, $h=l$, and using $(1-p')^w \sim e^{-12(\beta-\alpha)}$ we
immediately get 
$$p_\r{all~recs} \rightarrow 0 ~~~ \r{as} ~~~ n \rightarrow \infty.$$
\end{proof} %}}}

%{{{ Kus textu ktorym sa dorazi lower bound

\begin{figure} %{{{ obrazok
\thispagestyle{empty}
\caption{Left: Three consecutive stripes. In every (dashed) rectangle there
exists a
non $S'$ line (bold lines). These lines form path around the stripe. Right-top: Black
rectangle $R$ will expand itself to $ABCD$ because of incrementing vertices $v_i$.
Right-middle: Every square $2l \times 2l$ (dashed) w.h.p. contains this  
configuration of non $S'$ lines (bold lines). Bottom: results of a numerical study.}
\label{obrazok} 
\vbox{
  \hbox{
    \kern1cm
    \hbox{
      \includegraphics[width=0.25\textwidth]{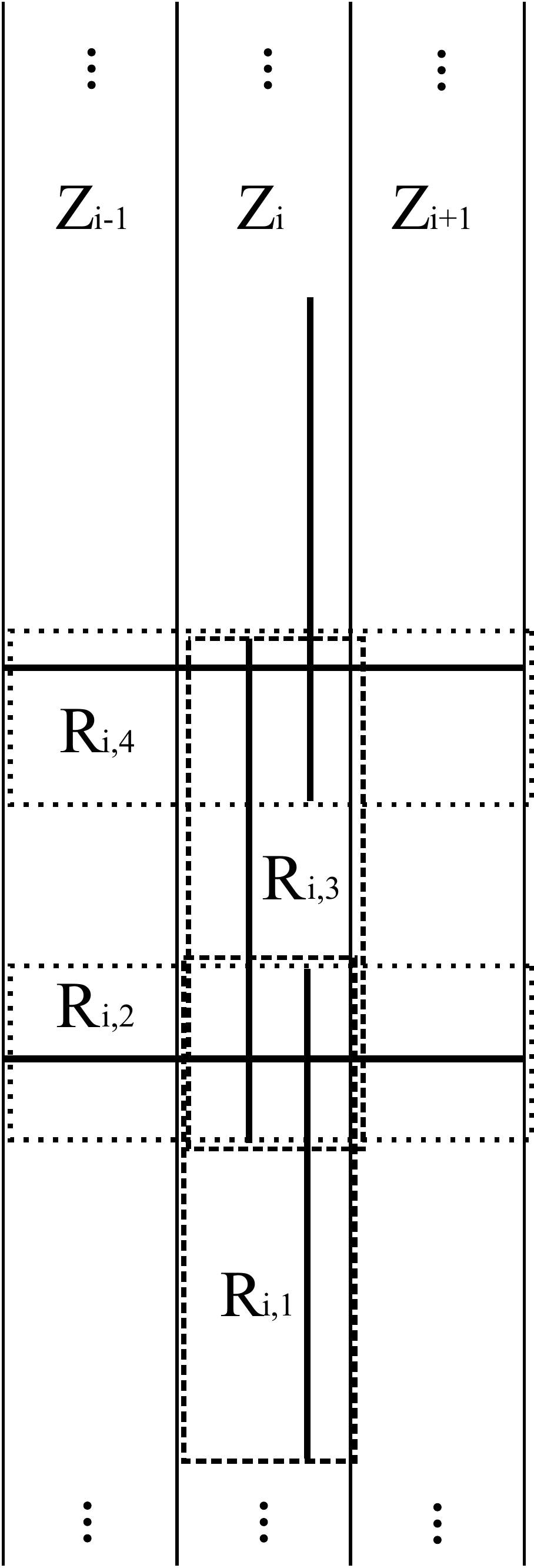} 
    }
    %\hfil
    \kern1.5cm
    \vbox{
      \hbox{
        \includegraphics[width=0.45\textwidth]{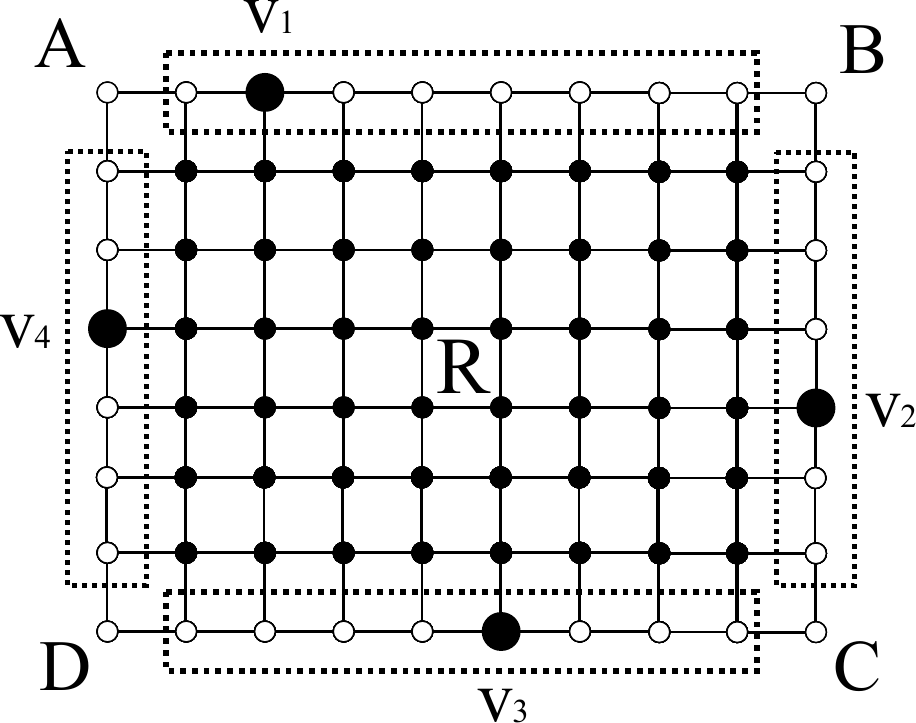} 
      }
      \kern1cm
      \hbox{
        \kern1cm
        \includegraphics[width=0.25\textwidth]{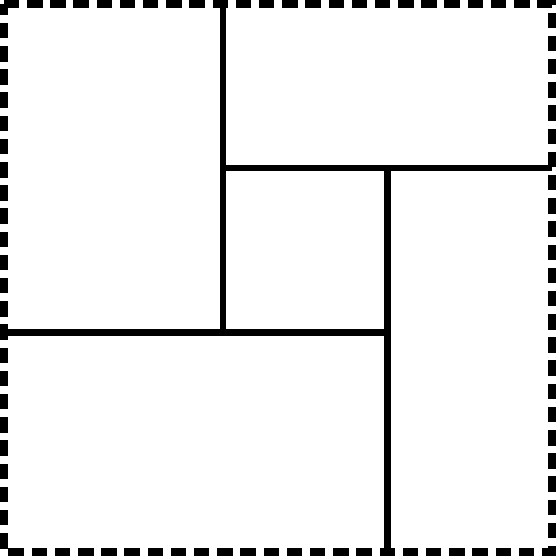} 
      }
    }
  }
  \includegraphics[width=0.90\textwidth]{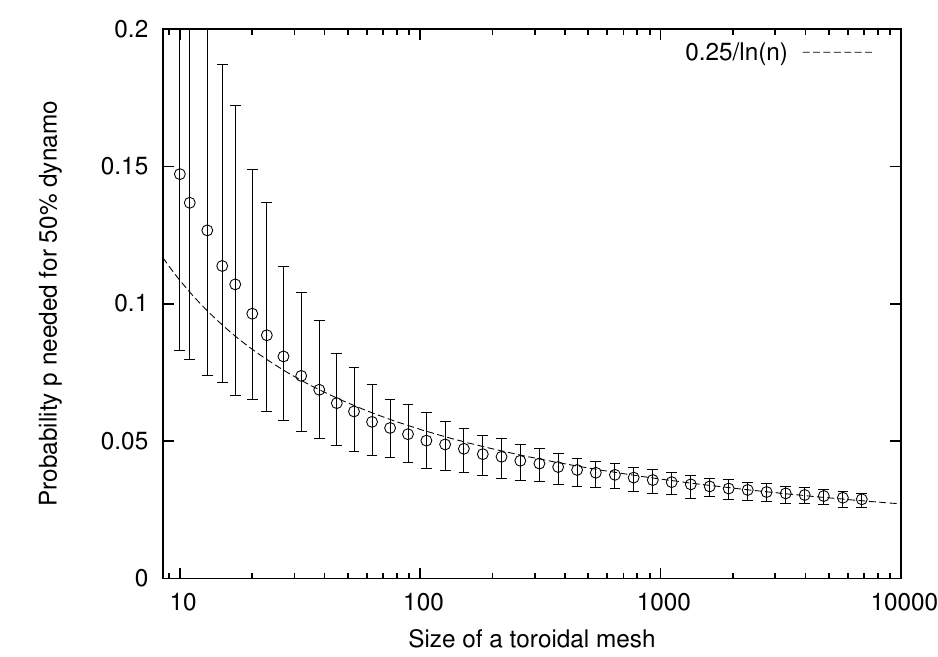} 
}
\end{figure}
%}}}

Now we show, how it is possible to satisfy preconditions of Lemma \ref{cage}.
Assuming that $n$ is divisible by $2l$ we divide $G'$ to $n/(2l)$ vertical stripes of
width $l$:
\begin{eqnarray*}
Z_0 &=& V'[0,0,l,n/2] \\  
Z_1 &=& V'[l-1,0,l,n/2] \\
Z_2 &=& V'[2l-1,0,l,n/2] \\
&\ldots&
\end{eqnarray*}
In each of these stripes we find closed path consisting of non $S'$ vertices 
that goes ''around'' this strip. This will be done as follows: In every stripe
$Z_i$ we construct a sequence of rectangles $R_{i,1}, R_{i,2}, \ldots$ as in the
top-left part of Figure \ref{obrazok}. Due to Lemma \ref{line} in each $R_{i,j}$
there is a non-$S'$-line $L_{i,j}$. Suitable parts of these lines form desired
path as is shown in the figure. Let us denote the
path that goes around the stripe $Z_i$ by $P_i$ and let $Q_i=\{ L_{i,j} | j \in
\mathbb{N}\}$. Let us now consider all horizontal lines from $Q_i \cup Q_{i+1}$.
These clearly divide the whole area between paths $P_1$ and $P_2$ to cages with
disjunct interiors. By repetition of this argument we can show that the area
between any two paths $P_i$ and $P_{(i+1) \bmod n}$ can be covered by cages with
respect to preconditions of Lemma \ref{cage}. For proof of Theorem \ref{not
black} it is therefore sufficient to calculate appropriate $\beta$ and $p$ that
would satisfy precondition of Lemma \ref{line}.

The very last thing needed to finish the proof of Theorem \ref{not black} is to deal
with such toroidal meshes that do not have its size divisible by $2l$. We shall
use the following intuitive lemma:

\begin{lemma}
\label{dlzka vplyvu}
Let $H$ be the toroidal mesh, $S_H$ the set of seeds, and let $t$ be such a time
after which no vertex changes its color to black. Let $S_H^*$ be another set of
seeds satisfying $S_H \subset S_H^*$ and let us denote $D:=S_H^*-S_H$. Then every vertex
that turns black in $\c{N}(S_H^*,H,2)$ after time $u>t$ is distant at most $2u$
from some vertex of $D$.  
\end{lemma} 

\noindent We assume that $n=2lk+r$ for some $k$ and $r<2l$. As before, we divide
$G$ to vertical stripes with width $2l$ (note that since we do not assume $n$
even as before, we are working with $G$ instead of $G'$) and one special stripe
with width $r$. As above, whole graph $G$ up to the special stripe can be
covered by cages. Every such a cage has the size of its interior limited to
$(2l) \cdot (6l) = 12l^2$. Therefore, if the special stripe contains only
non-$S^p$-vertices, then in time higher than $12l^2$ no vertex is turned black.
Therefore the preconditions of the Lemma \ref{dlzka vplyvu} are satisfied with
$H=G$, $D$ equal to the set of all $S^p$-vertices from the special stripe,
$S_H^*=S^p$, $S_H=S_H^* - D$ and $t=12l^2$. Therefore we know that every vertex
that turns black after time $t$ must lie within the distance $2t$ from some
black vertex from the special stripe. Note that dividing $G$ to the stripes can
be done in many ways - the special stripe can be located at arbitrary position.
If we assume that special stripe equals to $V[0,0,r,n]$, then every vertex that
turns black after time $t$ must has its $x$ coordinate within the interval
$\langle -2t,2t+r \rangle$. On the other hand, if the special stripe equals to
$V[\floor{n/2},0,r,n]$, then every vertex that turns black after time $t$ must
has its $x$ coordinate within the interval $\langle \floor{n/2}-2t,
\floor{n/2}+2t+r \rangle$. Since $t=o(n)$ and $r=o(n)$, these two intervals do
not overlap and therefore no vertex turns black in time $t+1$. This implies that
no vertex turns black anymore and therefore the initial coloring $S^p$ is not a
dynamo.
%}}}

\subsection{Upper bound} \label{section upper}%{{{
The result of this section is stated in the following theorem:

\begin{theorem} \label{black} 
If $p \geq 1.65/\ln(n)$, then the random set of seeds $S^p$ w.h.p. is a dynamo on graph $G$.
\end{theorem}

\noindent We want to illustrate the main idea of the proof first. The rigorous
proof is presented after few technical lemmas. Let $R$ be such square
that in some step of the coloring process consists of black vertices only. If in
each one of the four dashed lines (Figure \ref{obrazok} top-right) there exists at least
one seed (as are $v_1, v_2, v_3, v_4$ in the figure), then also the
square $ABCD$ will be black at some time. We shall call vertices $v_1, v_2, v_3,
v_4$ \emph{incrementing} with respect to $R$. Expansion of $R$ caused by incrementing
vertices can possibly continue until whole $G$ turns black. We shall show that we
can choose $R$ such that with probability tending to one this really will
happen.

Let $h=h(n)$ such that $h(n)$ is odd. Let us divide $G$ to the set of
$\floor{n/h}^2$ disjoint squares $U_i$ such that the size of $U_i$ is $h \times h$. The
event that $G$ turns black at some point in the coloring process is implied by
existence of specific square $U_k$ that satisfies:

\begin{itemize}
\item{$\Gamma_1$: }A square consisting of single vertex $v$ in the middle of
$U_k$ can grow due to the existence of incrementing vertices from
$S^p$ in such a way that in some time of the coloring process the
whole $U_k$ will be black. Note that the square of size $1 \times 1$ can grow
due to incrementing vertices to the square of size $3 \times 3$ even if it is
not black.

\item{$\Gamma_2$: }Square $U_k$ will grow due to the existence of incrementing
vertices from $S^p$, and at some time whole $G$ turns black. 
\end{itemize}

Let us note that although it is tempting to combine the conditions $\Gamma_1$ and
$\Gamma_2$, it can not be done easily. The reason is that for two 
different $U_i, U_j$, the corresponding $\Gamma_1$'s are independent
events. On the other hand, the ''spreading'' of the square $U_i$ to the whole $G$ is influenced
by (not) spreading of other squares $U_j$.  

We shall estimate probabilities of conditions $\Gamma_1$ and $\Gamma_2$ right after stating
following technical lemma.
%}}}

\begin{lemma} \label{zambezi} %{{{ 
Let $p>c/l$ for some constant $c$. Then every line of length $L:=l^3$ contains
w.h.p. at least one $S^p$-vertex.
\end{lemma}
\begin{proof}
The probability that the fixed line of length $L$ is non-$S^p$-line is 
$$p_\r{one~line}=q^L \sim e^{-pL} = e^{-c \cdot l^2}.$$
With the use of the union bound we get that the probability that at least one
line of length $L$ is non-$S^p$-line can be bounded as
$$p_\r{all~lines} \leq 2n^2 \cdot p_\r{one~line} \rightarrow 0.$$
\end{proof}%}}}

\begin{lemma} \label{zaire} %{{{
Let us assume that $h \rightarrow \infty$. Then the condition
$$\ln(n) - \frac{q \cdot \pi^2} {6p} - \ln(1/p) - \ln(h) \rightarrow \infty$$ 
w.h.p. implies that $\Gamma_1$ happens for arbitrarily large number of $U_k$'s. 
\end{lemma}

\begin{proof}
Let $A_i$ denote the event that for given $U_i$ the condition $\Gamma_1$ holds.
Let $X$ be the random variable that equals to the number of $A_i$'s holding.
There are at least $\floor{n/h}^2 \sim (n/h)^2$ different $U_i$'s and therefore
the expected value of $X$ satisfies the relation $\E(X) \sim (n/h)^2 \P(A_i)$.
Since the $A_i$'s are independent, $\E(X) \rightarrow \infty$ implies that
w.h.p. at least one of the $A_i$'s holds. Let us estimate $\P(A_i)$ (for fixed
$U_i$) as follows: 
\begin{equation}
\label{rovnica} 
\P(A_i) \geq V_{\rm odd} \quad {\rm where} \quad V_{\rm odd}=\prod_{i \geq 1, i {\rm~is~odd}} (1-q^i)^4  
\end{equation}
Furthermore we define
$$V_{\rm even}=\prod_{i \geq 1, i {\rm~is~even}} (1-q^i)^4 \quad {\rm and} \quad 
V_{\rm all}=V_{\rm even} \cdot V_{\rm odd}$$

Note that $V_{\rm odd} < V_{\rm even}$ but $V_{\rm odd} > V_{\rm even} \cdot
(1-q)^4=V_{\rm even}\cdot p^4$. Therefore 
$V_{\rm odd} \geq \sqrt{V_{\rm all}} \cdot p^2$. Further, we compute:

\begin{eqnarray}
V_{\rm all} &=& \prod_{i=1}^\infty (1-q^i)^4 = \exp \Bigl (4\sum_{i=1}^\infty \ln(1-q^i) \Bigr ) = 
\exp \Bigl (-4 \sum_{i=1}^\infty q^i + q^{2i}/2 + q^{3i}/3 + \ldots \Bigr ) = 
\nonumber \\
&=& \exp \Bigl (-4 \sum_{i=1}^\infty \sum_{k=1}^\infty q^{ki}/k \Bigr )= 
\exp \Bigl (-4 \sum_{k=1}^\infty \sum_{i=1}^\infty q^{ki}/k \Bigr )= 
\nonumber \\
&=& \exp \Bigl (-4 \sum_{k=1}^\infty \frac{1}{k} \cdot \frac{q^k}{1-q^k} \Bigr )= 
\exp \Bigl (-4 \sum_{k=1}^\infty \frac{1}{k} \cdot \frac{1}{q^{-k}-1} \Bigr )
\nonumber
\end{eqnarray}
where in the second step we used Taylor's expansion for $\ln(1+x)$.
Furthermore, we use Bernoulli's inequality:

$$q^{-k}-1=(1+z)^k-1 \geq kz$$

\noindent where we put $1+z = 1/q$. Using this we obtain
\begin{eqnarray*}
\ln(V_{\rm all})/4 \geq -\sum_{k=1}^\infty \frac{1}{zk^2}= 
-\frac{1}{z}\sum_{k=1}^\infty \frac{1}{k^2}= 
-\pi^2/6 \cdot \frac{1}{z}= -\pi^2/6 \cdot \frac{q}{p}.
\end{eqnarray*}

\noindent Finally $\E(X) \rightarrow \infty$ can be written as:
$$(n/h)^2 \cdot \P(A_i) \geq (n/h)^2 \sqrt{\exp \Bigl (-4 \cdot \frac{\pi^2 \cdot q}{6p}
\Bigr )} \cdot p^2 \rightarrow \infty$$ 
Calculating logarithm of the last relation proves the lemma.
\end{proof} %}}}

\begin{proof} (Theorem \ref{black}). %{{{ 
Let us choose $h=2l^3+1$ and let $p$ be as in the statement of Theorem
\ref{black}. The direct calculation shows that the preconditions of Lemma
\ref{zaire} are satisfied, and therefore there a.s. exists $U_k$ that satisfies the
condition $\Gamma_1$. Moreover, since $h>l^3$, Lemma \ref{zambezi} implies that
there w.h.p. exist all incrementing vertices needed for the condition
$\Gamma_2$ to hold.  
\end{proof}
%}}}

% poznamka o moznych zlepseniach {{{
Let us note that both Theorems \ref{black} and \ref{not black} can be sharpened
a little. When improving Theorem \ref{not black} we can show (with some
additional effort) that for $p \leq 0.018/\ln(n)$ every square of size $2l \times
2l$ a.s. contains four non-$S^p$-lines such as in the Figure \ref{obrazok}
(middle-right). Existence of such structures is sufficient for showing the
existence of covering of $G$ by cages as required by Lemma \ref{cage}. For
improvement of Theorem \ref{black}, let us note that we do not actually need all
incrementing vertices whose existence is included in (\ref{rovnica}) in Lemma
\ref{zaire}. For successful ''spreading'' of black vertices it is sufficient if
only approximately half of them exist. Let us recall that in Lemma \ref{zaire} we
started with $1 \times 1$ square and this square potentially grew until the whole
$S_k$ became black. The expansion of this square of size $i \times i$ to $(i+2) \times
(i+2)$ requires the existence of the set of four incrementing vertices $U_1$.
Another set of incrementing vertices $U_2$ then causes expansion to size $(i+4)
\times (i+4)$. But in almost all situations, $U_2$ can cause expansion from $i
\times i$ to $(i+4) \times (i+4)$ directly - without need of $U_1$. This leads
to $p \geq 0.83 / \ln(p)$.

This considerations sharpen the results, but it is still an interesting question, whether there
exists a threshold $t$ such that if $t'>t$ and $p>t'/\ln(n)$, then the initial
random set of seeds is a dynamo and vice-versa. Sadly, we are not able to prove
the (non)existence of such threshold and determine its possible value. Therefore we
present a numerical study that tries to answer these questions at least
partially. By $p_z(n)$ let us denote the probability $p$, such that the random
set of seeds $S^p$ is a dynamo with probability $z$. Results of the
numerical study are presented in the Figure \ref{obrazok} (bottom). The circles
present values of $p_{0.5}(n)$. The top and the bottom points of the
''error bars'' correspond to values $p_{0.05}(n)$ and $p_{0.95}(n)$. Size of
$90\%$ confidence intervals for measured values is comparable with the diameter
of the circles.

We conclude that threshold $t$ is likely to exist and its value is near $0.25$.
Proving this statement rigorously remains an open problem.
%}}}

%{{{ 2. cast uvod
\section{Dynamos on random $4$-regular graph} \label{sec:regular}
In the previous section we analyzed a particular $4$-regular graph. In this
section we try to solve a similar task for a random $4$-regular graph under
simple majority scenario. 

Random regular graphs are very well studied. Pioneering studies in this field
were brought by Bender and Canfield \cite{11}, Bollob\' as \cite{12} and Wormald
\cite{13,14}. A systematic research in this area grew enormously since then,
partly driven by applications in many areas
such as computer science. For exhaustive survey see \cite{15}. To make this
paper self-contained, we now present the definition and basic properties that
are necessary for further reading.

We use $\c{G}_{n,d}$ to denote the uniform probability space of $d$-regular
graphs on the $n$ vertices $\{1, 2, \ldots, n\}$ (where $dn$ is even). 
Sampling from $\c{G}_{n,d}$ is therefore equivalent to taking such a graph uniformly at
random (u.a.r.). Another possible way how to define (the same) probabilistic
space is this: We construct graph with $n$ vertices (named as before) and no
edges. Then, for every vertex $i$ we construct $d$ \emph{slots} $v_{i,1}, \ldots,
v_{i,d}$. A perfect matching of these slots into $dn/2$ pairs is called a
\emph{pairing}. A pairing $P$ corresponds to a multigraph (with loops and
multiple edges
permitted), in which two vertices $i$ and $j$ are connected, if there exist
two different slots $v_{i,x}$ and $v_{j,y}$ that form a pair from $P$. Although this
process can lead to graph that is not simple, it is quite easy to show that
the probabilities of obtaining two simple graphs $G_1$ and $G_2$ are equal.
Therefore if we reject every graph that is not simple and repeat the whole process
until simple graph is found, we obtain the same probabilistic space
$\c{G}_{n,d}$ as before. 

A pairing can be selected u.a.r. in many different ways. In particular, the
slots in the pairs can be chosen sequentially. At any stage, the first slot in
the next random chosen pair can be selected using any rule whatsoever, as long
as the second slot in that pair is chosen u.a.r. from the remaining slots. For
example, one can insist that the next chosen slot is the next one available in
any pre-specified ordering of the slots, or comes from a cell containing one of
the slots in the previous chosen pair (if any such slots are still unpaired).
We shall call this the \emph{independence property} of the pairing model.

In what follows $G \in \c{G}_{n,4}$ will be random $4$-regular graph and
$l=\floor{\ln(n)}$ as before. Also, we consider the same coloring process as the
one described in preliminaries to section \ref{sec:torus}.

Let us note that we do not need to restrict our interest to simple graphs. The
coloring process is well-defined also for multi-graphs. Bender and Canfield
showed that the probability of obtaining a simple graph by the pairing-slot
process is asymptotically $\exp((1-d^2)/4)$ ($d$ being the degree of a vertex),
which is a value close to $2\%$ for $4$-regular graph and high number of
vertices. Therefore if we prove  that $S^p$ is not a dynamo w.h.p. on random
$4$-regular multi-graph, it implies that $S^p$ is not w.h.p. dynamo on random
$4$-regular graph.

%}}}

In the future text, let $G$ be random $4$-regular (simple) graph of size $n$ and
$H$ be random $4$-regular (not necessarily simple) graph of size $n$.

\begin{theorem} \label{2cast black} %{{{ 
Let $p \geq 0.12$. Then $S^p$ w.h.p. is a dynamo on $G$. 
\end{theorem}

\begin{proof}
 As was justified above, we shall work with $H$ instead of $G$.
We use some properties of the random regular graph that are summarized in
\cite{15}. The graph has only small number of short cycles. Therefore, the
neighborhood of any vertex looks like a part of an infinite tree. For $v \in
V(H)$, let $T_v$ be a rooted tree that is obtained as follows: Vertex set of $T$
is a subset of $V(H)$, vertex $v$ is root of $T$. Tree $T$ consists of
$h=\floor{2 \cdot \lg \log(n)}$ \emph{levels}. Vertex $v$ forms $0$-th level of
the tree. The $(i+1)$-th level consists of all vertices that are neighbors (in
$H$) to some vertex in the $i$-th level and that are not already contained in
the  $j<i$-th level for any $j$. If $u, v$ are two vertices of $i$-th and $(i+1)$-th
level respectively, then there exists an edge between $u$ and $v$ in $T$ if
there exists an edge between $u$ and $v$ in $H$ and furthermore, there is no
other vertex $w$ than $u$ from $i$-th level of $T$ that is adjacent to $v$ (in
$H$). There are no edges within one level of $T_v$. Note that although
some edges from $H$ are not present in $T_v$, with the use of the independence
property we can easily calculate that the probability that some vertex in $i$-th
level ($i<h$) has degree lower than $4$ is $o(1)$. 

It is clear that if we run the coloring process on $T_v$ instead of $H$ (but we
still set the threshold $\psi$ to $2$ for all vertices), then the probability
that the root of $T_v$ turns black in coloring process induced by $S^p(T_v)$ is
lower than the probability that the same vertex $v$ turns black in $H$ (in the
coloring process induced by $S^p(H)$). By $Y(i)$ let us denote the probability
that the vertex of the $i$-th level of $T_v$ turns black (in the coloring
process induced by $S^p(T_v)$). We shall start analyzing the bottom level of the
tree and proceed to the top. It obviously holds that $Y(h)=p$. For the
$(l-i)$-th level ($i>0$) the situation is more complicated. In order to obtain
some estimate for $Y(l-i)$ we shall assume that the color of some vertex $u$
from $(l-i)$-th level can be changed only because of (at least) two black
neighbors of $u$ belonging to $(l-i+1)$-th level. This gives us the following
estimation:
\begin{equation}
\label{rovnicaaa}
Y(i-1) \geq p + (1-p) \cdot f_{3,\geq 2}(Y(i))
\end{equation}
where
$$f_{n,=k}(p)=\binom{n}{k} \cdot p^k \cdot (1-p)^{n-k} \quad {\rm
and} \quad f_{n,\geq k}(p)=\sum_{i=k}^n f_{n,=i} $$

\noindent Numerical calculation shows that given $p \geq 0.12$ the condition
$Y(h-100) \geq 0.999$ holds. Here we
change our approach and analyze higher levels analytically. We easily prove that
for any $p$, $q=1-p$, it holds that:

\begin{equation}
\label{rovnicabb}
1-f_{3,2}(p) \leq 5q^3
\end{equation}

\noindent Let us denote $\varepsilon_{i}=1-Y(i)$. Then we
can rewrite (\ref{rovnicaaa}) using (\ref{rovnicabb}) to the form:
$$Y(i-1) \geq p + (1-p) \cdot (1-5 \cdot \varepsilon_{i}^3)$$
And finally, we get: $\varepsilon_{i-1} \leq 5 \varepsilon_{i}^3.$
Last inequality together with numerically obtained result
$\varepsilon_{h-100} < 0.001$ implies, that
$$\log(\varepsilon_{i-1}) \leq \log(5) + 3 \log(\varepsilon_{i}) \leq 2
\log(\varepsilon_i).$$
Therefore 
$$\log(\varepsilon_{0}) \leq -(2^{h-\Theta(1)}) = -\log(n)^2/\Theta(1)
\leq -2\log(n)$$ for sufficiently big $n$ and finally $\varepsilon_{0} \leq n^{-2}.$
This means that the expected value of number of vertices from $H$ that do not turn
black is $o(1)$.
\end{proof}
%}}}

\begin{theorem} \label{2cast not black} %{{{
Let $p \leq 0.10$. Then $S^p$ w.h.p. is not a dynamo on $G$. 
\end{theorem}

\noindent Before we present the actual proof of this theorem, let us recall several facts 
concerning the \emph{Galton-Watson branching process}.

Let $Z$ be a probability distribution on the nonnegative integers. The
\emph{Galton-Watson branching process} with offspring distribution $Z$ is
(loosely) defined as follows. The $0$-th generation consists of $N$
particles. For $t\geq 0$, the $(t+1)$-th generation consists of the
children of all particles in the $t$-th generation. For each particle,
the number of its children has distribution $Z$, and is independent of all the
other particles. 

Let $X$ be a (Galton-Watson) branching process with $N$
particles in the beginning, distribution $Z$ and we shall denote the number of
particles in generation $t$ by $X_t$ and the average number of children 
of a particle by $\lambda~(=\E(Z))$. Moreover, let 

$$X_{\rm sum}:=\sum_{i=0}^\infty X_t.$$

\begin{lemma} 
\label{branching}
For all $t\geq 0$, 
$$\E(X_t) = N \cdot \lambda^t \quad {\rm and~if~}\lambda<1{\rm~then}  \quad \E(X_{\rm sum}) =
\frac{N}{1-\lambda}.$$
Furthermore, if $\lambda<1$, and $\sigma^2$ is the finite variance of $Z$, then
$$\P(X_{\rm sum} > 2 \cdot \E(X_{\rm sum})) = o(1).$$
\end{lemma}

\begin{proof}
The proof of the first and the second statement is trivial. For the proof of the last
statement, we use the first part of the lemma and variance estimation $\var(X_{\rm sum})
\leq N \sigma^2 / (1-\lambda)^2$ (see for instance \cite{10}). Therefore
Chebyshev inequality is applicable, leading to the required result. 
\end{proof}

\noindent Now we are ready to prove the main result of this section.

\begin{proof} (Theorem \ref{2cast not black})
 As was justified above, we shall work with $H$ instead of $G$.
Although an approach quite different from that used in previous section is to be
used for proving the Theorem \ref{2cast black}, there is one similarity: we
first numerically compute how the coloring process evolves (w.h.p.) up to the
certain time and only then we solve the problem analytically by finding analogy
with easily solvable branching process.

In Table \ref{tab:algoritmus} we present a randomized algorithm that simulates
the coloring process on random $4$-regular multi-graph. The algorithm returns
$1$ with the same probability as is the probability that $S^p$ is a dynamo on
random $4$-regular multi-graph. Let us recall that the analysis of the coloring
process on multi-graphs is sufficient for proving the statement of this theorem. 

The main idea behind the algorithm is that the edges in $H$ are not sampled at
once. Instead, in the $i$-th step of the coloring process we sample only edges
corresponding to slots that belong to black vertices and have not yet been
sampled (that is, $B_i$). By doing this, the sampling of the graph $H$ takes
place ''in the same time'' as the black color spreads. Because of independence
property, this leads to the same probabilistic space of coloring processes as if
we first sample the whole graph $H$, then the random initial coloring and only
then start with the coloring process.

The meaning of used symbols is: $B_i$ (as defined in the preliminaries) is
the set of vertices that turns black in the $i$-th step of the coloring process;
$M_i$ is the set of white vertices that have exactly one black neighbor in the
$i$-th step of the coloring process (excluding $B_i$); $T_i$ are neighbors of $B_i$
and they are split up to the sets $T_i^k$ ($k=1,2,3$) according to their
future contribution to one of following sets: $B_{i+1}, B_{i+1}, M_{i+1}$. 

\begin{table}
\label{tab:algoritmus}
\caption{Sampling-Coloring (SC) Algorithm}
\begin{tabular}{|p{0.05\textwidth}|p{0.85\textwidth}|}
\hline \hline
1. & Sample the set $B_0$ and let $R_0:=V(H)-B_0$. \\
%\hline
2. & Sample the neighbors of all vertices from $B_0$. Let $M_1$ be those
vertices from $R_0$ that have one neighbor in $B_0$ and let $B_1$ be those
vertices that have two or more neighbors. Let $R_1:=R_0-B_1-M_1$. \\
%\hline
3. & Let $i:=1$. \\
%\hline
4. & Sample the neighbors of $B_i$ and denote by $T_i$ the set of them. \\
%\hline
5. & Let $T_i^1 \subseteq T_i \cap R_i$ be the vertices that have at least two
neighbors in $B_i$. \\
%\hline
6. & Let $T_i^2 := T_i \cap M_i$. \\
%\hline
7. & Let $T_i^3 := T_i-T_i^1-T_i^2$. \\
%\hline
8. & Let $B_{i+1}:=T_i^1 \cup T_i^2$. \\
%\hline
9. & Let $M_{i+1}:=(M_i - T_i^2) \cup T_i^3$. \\
%\hline
10. & Let $R_{i+1}:=R_i-B_{i+1}-M_{i+1}$. \\
%\hline
11. & Let $i:=i+1$. \\
%\hline
12. & If $B_i=\emptyset$ and $R_i=\emptyset$ return 1; if $B_i=\emptyset$ and
$R_i \neq \emptyset$ return 0.\\ 
%\hline
13. & Go to line 4. \\ \hline \hline 
\end{tabular} 
\end{table} 

We now calculate some estimations for the cardinalities of $B_i, R_i$ and
$T_i^k$. For abbreviation we shall use $x$ for $\abs{X}$ for all used sets. We
shall restrict the considerations to $i \leq 60$ and $p=0.10$. Note, that under
these assumptions all cardinalities are $\Theta(n)$. We define $z_i$ as the
ratio of number of empty slots of the vertices from $B_i$ to the total number of empty
slots in the $i$-th step, that is:

\begin{equation}
\label{rov:z}
z_i \leq \frac{2 \cdot b_i} {4 \cdot r_i + 3 \cdot m_i + 2 \cdot b_i}
\end{equation}

\noindent for $i>0$ and $z_0=|B_0|/|V(H)|$. By $f_{n,\geq k}$, $f_{n,=k}$ we
shall mean the functions defined in the proof of Theorem \ref{black}.  Now the
following estimations hold:

\begin{lemma}
For abbreviation, let $X:=t_k^1$. Then, if $r_k,b_k=\Theta(n)$, then  
$$\E(X)=r_kf_{4,\geq 2}(z_k) (1+o(1)) \quad {\rm and} \quad \var(X) < 
E(X) \Theta(1) + \E^2(X) o(1).$$
\end{lemma}

\begin{proof}
Let $X=\sum_i X_i$ where $X_i$ is an indicator random variable indicating that
$i$-th vertex from $R_k$ has at least two neighbors from $B_k$, that is, it
belongs to $T_k^1$ (and consequently to $B_{k+1}$). The probability that this
happens (that is, $X_i=1$) is clearly $f_{4,\geq 2}(z_k) (1+o(1))$. There are
$r_k$ such vertices, which proves the first part of the lemma. For the second
part, let us compute:

$$\var(X)=\E(X^2)-\E^2(X)=\sum_{i,j}\E(X_iX_j)-\E^2(X)= \sum_{i \neq j}
\E(X_iX_j) + \E(X) - \E^2(X).$$
Clearly for $i \neq j$ holds $\E(X_iX_j)=(f_{4,\geq 2}(z_k))^2 (1+o(1))$. There
are $r_k(r_k-1)$ such terms in the sum. The second part of the
lemma follows by straightforward calculation.
\end{proof}

From this lemma and from Chebyshev inequality we have, that if $r_k,b_k=\Theta(n)$,
then with high probability $t_k^1=r_kf_{4,\geq 2}(z_k) \gamma$, where $\gamma$ is some
constant that is arbitrarily close to $1$. Quite similarly, we can derive such
relations for cardinalities of all sets of SC algorithm. This gives us:

\begin{tabular}{p{0.4\textwidth}p{0.4\textwidth}}
\begin{eqnarray}
b_0 &=& p \cdot n \cdot \gamma \nonumber \\
r_0 &=& n - b_0 \nonumber \\
b_1 &=& r_0 \cdot f_{4,\geq 2} (z_0) \cdot \gamma \nonumber \\
m_1 &=& r_0 \cdot f_{4,=1} (z_0) \cdot \gamma \nonumber \\
r_1 &=& r_0 - b_1 - m_1 \nonumber \\
t_i^1 &=& r_i \cdot f_{4,\geq 2}(z_i) \cdot \gamma \nonumber 
\end{eqnarray}
&
\begin{eqnarray}
t_i^2 &=& m_i \cdot f_{3,\geq 1}(z_i) \cdot \gamma \nonumber \\
t_i^{'2} &=& m_i \cdot f_{3, \geq 1}(z_i) \cdot \gamma \nonumber \\
t_i^3 &=& r_i \cdot f_{4,=1}(z_i) \cdot \gamma \nonumber \\
b_{i+1} &=& t_i^1 + t_i^2 \nonumber \\
m_{i+1} &=& m_i+t_i^3-t_i^{'2} \nonumber \\
r_{i+1} &=& r_i-b_{i+1}-(m_{i+1}-m_i) \nonumber 
\end{eqnarray}
\end{tabular}

\noindent For proving that the coloring does not form a dynamo we need to upper
bound $m_i$ and $b_i$ and lower bound $r_i$. Therefore we put $z_i$ as big as
possible (see (\ref{rov:z})), and furthermore we put $\gamma=1.0001$ in all
equations above except for the case of $t_i^{'2}$, where we use $\gamma=0.9999$.
With the use of this approach we obtain the following results:

$$m_{60} < 0.36n \quad b_{60} < 10^{-8}n \quad r_{60} > 0.41n \quad
\Big |\bigcap_{i=0}^{60}B_i \Big | <0.24n$$

\noindent The statement of the theorem looks now very persuasive - the fraction
of vertices that turn black in the $60$-th step of the coloring process is
lower than $10^{-8}$. However, we can not continue with evaluation of the above 
equations indefinitely, since we need all quantities to be at least $\Theta(n)$.
Otherwise, it would be hard to bound $\gamma$ present in these equations.
However, we can use the analogy with the branching process discussed by Lemma
\ref{branching}. Let us now consider the following conditions: 

\begin{equation}
\label{podmienky}
m_k > 0.37n \quad b_k > 10^{-6}n \quad r_k < 0.40n \quad
\Big |\bigcap_{i=0}^kB_i \Big | >0.25n
\end{equation}

\noindent We prove that w.h.p. there is no such $k$ that would satisfy a
condition from (\ref{podmienky}). We already concluded that this is true for $k
\leq 60$. Let us fix some coloring process on the graph $H$ such that at least
one condition from (\ref{podmienky}) holds and let $k'$ be the minimal $k$ for
which this happens. It is easy to calculate that until no condition from
(\ref{podmienky}) holds, the expected number of vertices that turn black in the
$(60+k)$-th step ($k \geq 0$) can be upper bounded by the expected
number of particles that exists in the $k$-th generation of branching
process with parameters $N=10^{-8}n$ and $\lambda=0.9$. This gives us the
equality $\E(X_{\rm sum})=10^{-7}$. Note also that any of the conditions from
(\ref{podmienky}) implies that the number of the particles in the branching
process exceeds $100N=10 \cdot \E(X_{\rm sum})$, which by Lemma \ref{branching}
happens with probability $o(1)$. Therefore, w.h.p. less than $40\%$
from all vertices turn black. 
\end{proof} 

Finally, there are rounding errors in our numerical analysis. Influence of
these errors can be estimated by multiplication (or division) of the result of
each numerical operation by some number close to $1$. In the proof of Theorem
\ref{2cast not black} this is clearly hidden in our choice of $\gamma$; in the
proof of Theorem \ref{2cast black} we can simple multiply the right side of
(\ref{rovnicaaa}) by $1-10^{-6}$ which produces more significant ''error'' than
the error of the machine. However, the result $Y(h-100) \geq 0.999$ still holds.

%}}}

%bibliografia {{{

\end{document}